\documentclass[conference]{IEEEtran}
\usepackage{amsmath}
\usepackage{amsthm}
\usepackage{amssymb}
\usepackage{bm}
\usepackage{xspace}
\usepackage{xcolor}
\usepackage{graphicx}
\usepackage{url}
\usepackage{cite}
\usepackage{framed}
\usepackage{float}
\usepackage{rotating}

\newcommand{\executeiffilenewer}[3]{%
\ifnum\pdfstrcmp{\pdffilemoddate{#1}}%
{\pdffilemoddate{#2}}>0%
{\immediate\write18{#3}}\fi%
}

\graphicspath{{images/}}

\newcounter{algocount}

\newenvironment{algorithm}[1][]{\refstepcounter{algocount}\begin{trivlist}\item \textbf{Algorithm \thealgocount.}#1\\[-0.2cm]\rule{\columnwidth}{1pt}}{\\[-0.2cm]\rule{\columnwidth}{1pt}\end{trivlist}}

\theoremstyle{plain}

\newtheorem{proposition}{Proposition}

\newtheorem{lemma}{Lemma}

\theoremstyle{definition}

\theoremstyle{plain}

\theoremstyle{definition}

\theoremstyle{remark}

\newcommand{\de}{\,\mathrm{d}}

\newcommand{\vecone}{\boldsymbol{1}}

\newcommand{\vecmu}{\boldsymbol{\mu}}

\newcommand{\vecd}{\boldsymbol{d}}

\newcommand{\vecp}{\boldsymbol{p}}

\newcommand{\vect}{\boldsymbol{t}}

\newcommand{\vecv}{\boldsymbol{v}}
\newcommand{\vecw}{\boldsymbol{w}}

\DeclareMathOperator{\kl}{D}

\DeclareMathOperator{\entop}{\mathrm{H}}

\DeclareMathOperator*{\minimize}{minimize}

\DeclareMathOperator*{\st}{subject\,to}

\title{Writing on the Facade of RWTH ICT Cubes:\\
Cost Constrained Geometric Huffman Coding}

\IEEEoverridecommandlockouts

\author{\IEEEauthorblockN{Georg B\"ocherer\IEEEauthorrefmark{1}, Fabian Altenbach\IEEEauthorrefmark{1}, Martina Malsbender\IEEEauthorrefmark{2}, and Rudolf Mathar\IEEEauthorrefmark{1}}
\IEEEauthorblockA{\IEEEauthorrefmark{1}Institute for Theoretical Information
Technology\\
RWTH Aachen University, 52056 Aachen, Germany
\\ Email: \texttt{\{boecherer,altenbach,mathar\}@ti.rwth-aachen.de}}
\IEEEauthorblockA{\IEEEauthorrefmark{2}kadawittfeldarchitektur gmbh
\\
52064 Aachen, Germany
\\ Email: \texttt{ICTcubes@kadawittfeldarchitektur.de}}
\thanks{This work has been supported by the UMIC Research Center, RWTH
Aachen University.}
}

\begin{document}
\maketitle

\begin{abstract}
In this work, a coding technique called \emph{cost constrained Geometric Huffman coding} (\textsc{ccGhc}) is developed. \textsc{ccGhc} minimizes the Kullback-Leibler distance between a dyadic probability mass function (pmf) and a target pmf subject to an affine inequality constraint. An analytical proof is given that when \textsc{ccGhc} is applied to blocks of symbols, the optimum is asymptotically achieved when the blocklength goes to infinity. The derivation of \textsc{ccGhc} is motivated by the problem of encoding a text to a sequence of slats subject to architectural design criteria. For the considered architectural problem, for a blocklength of $3$, the codes found by \textsc{ccGhc} match the design criteria. For communications channels with average cost constraints, \textsc{ccGhc} can be used to efficiently find prefix-free modulation codes that are provably capacity achieving.
\end{abstract}

\section{Introduction}

In the near future, parts of the electrical engineering faculty of RWTH Aachen University will move into new buildings called \emph{Information and Communication Technology} (ICT) \emph{cubes}. To protect the cubes against heating up in sun light, the idea is to shadow the facades by placing rows of slats in front of them. The slats itself come in three forms, left, right, and middle. All slats types have a height of $1.70$m. The widths are given by $0.18$m, $0.18$m, and $0.31$m, respectively. Each $0.625$m a slat is placed. See also Fig.~\ref{fig:visualization} for a visualization of the cubes. To cover all eight facades of the two cubes, a total number of $4264$ slats is required. The actual choice of slats is subject to the following design criteria.
\begin{enumerate}
\item[C1.] For aesthetic reasons, the sequence of slats should appear random.
\item[C2.] To ensure enough cooling, around 33\% of the facade area should be covered by the slats.
\item[C3.] Since shadow turns the rooms dark, the total shadowing should not exceed 33\%.
\end{enumerate}
Observing that many different sequences of slats fulfill the above constraints, Mr. Mathar came up with the idea to encode a text to the sequence of slats, when read row by row from left to right. Thus, the challenge is to encode a text to a sequence of slats subject to the design criteria C1., C2, and C3.
\begin{figure}
\includegraphics[width=\columnwidth]{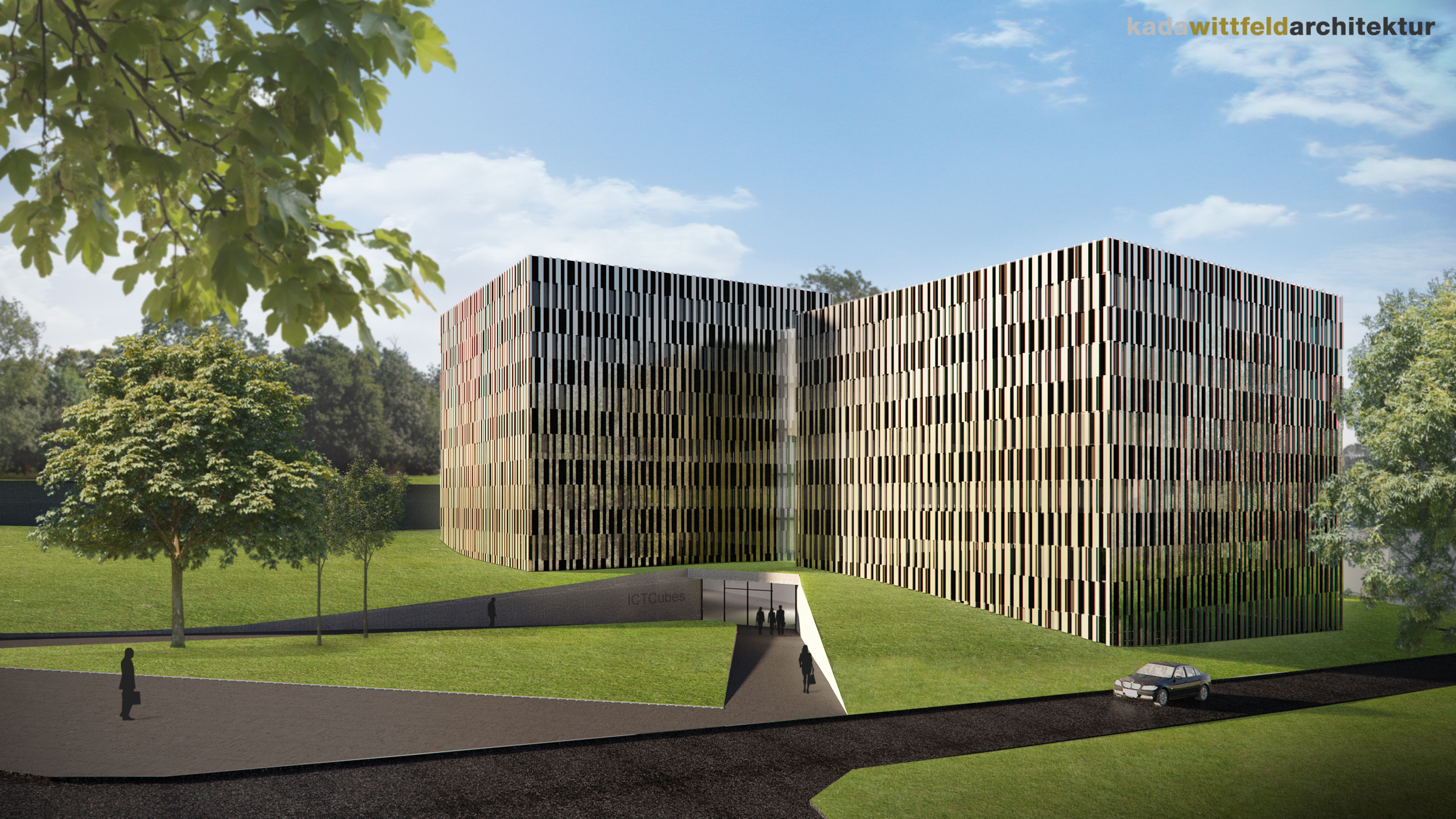}
\caption{Visualization of the ICT cubes.}
\label{fig:visualization}
\end{figure}

In the remainder of this work, we develop a coding scheme that solves this problem. The key part of our scheme is a new algorithm that we call \emph{cost constrained geometric Huffman coding} (\textsc{ccGhc}). This algorithm minimizes the Kullback-Leibler distance between a dyadic probability mass function (pmf) and a target pmf subject to an affine inequality constraint. Interestingly, in the context of channel matching \cite{Bocherer2011}, \textsc{ccGhc} can also be used to directly find capacity-achieving modulation codes for communication channels with average power constraint. This improves upon the broad search approach that we presented in \cite{Bocherer2011a}.

\section{Approach}
\subsection{Problem Modelling}
As stated in the introduction, there are three types of slats, i.e., left, right, and middle slats. We index them in this order. To turn the design problem into a tractable problem, we use a probabilistic model. Assume each slat is drawn independent and identically distributed (iid) from the set $\{1,2,3\}$ according to a pmf $\vecp=(p_1,p_2,p_3)^T$. According to criterion C1., we would ideally choose uniformly among the three types of slats. Thus, we would like the pmf $\vecp$ to be close to the uniform target pmf $\vect=(1/3,\,1/3,\,1/3)^T$. As a distance measure, we use the Kullback-Leibler (KL) distance, which is defined as
\begin{align}
\kl(\vecp\Vert\vect)=\sum_i p_i\log\frac{p_i}{t_i}.
\end{align}
Thus, criterion C1. can be cast into the objective to minimize $\kl(\vecp\Vert\vect)$. The width of the slats in meters is
\begin{align}
\vecw=(w_1,w_2,w_3)^T=(0.18,0.18,0.31)^T\quad[\mathrm{m}].
\end{align}
Every $0.625$m, a slat is placed. Thus, by criterion C2., each slat has to cover in the average a 
breadth of 
\begin{align}
S=33\%\cdot 0.625 = 0.2063.
\end{align}
Note that $\vecw^T\vect/0.625\approx 36\%$, i.e., when using the uniform distribution, the shadowing is too strong and criterion C3. is violated. Thus, criterion C2. and C3. can be cast into the affine inequality constraint $\vecw^T\vecp\leq S$. 

Pmfs of the slats can be generated as follows. We do source-channel separation with a binary interface, i.e., we first compress the text to a binary sequence, and we then design a code that maps the binary sequence to a sequence of slats with the objective to match the design criteria. The text compression part is a well-studied topic. For now, we therefore assume perfect compression, i.e., after text compression, we have a stream of iid equiprobable bits. By parsing the binary stream by a full prefix-free code, we can generate dyadic pmfs $\vecd$ \cite{Bocherer2011}, i.e., pmfs where each entry $d$ is of the form
\begin{align}
d=2^{-\ell},\qquad \ell\in\mathbb{N}. 
\end{align}
Thus, our objective is to approximate the target pmf $\vect$ by a dyadic pmf $\vecd$ while guaranteeing in the average a shadowing of at most $S$. Within the probabilistic model, the criteria C1.-C3. can now be cast into the following optimization problem.
\begin{align}
\begin{split}
\minimize_{\vecd}\quad&\kl(\vecd\Vert\vect)\\
\st\quad&\vecw^T\vecd\leq S\\
&\vecd \text{ is a dyadic pmf}.
\end{split}
\label{prob:dyadic}
\end{align}
\subsection{Cost Constrained Geometric Huffman Coding}
Without the restriction of pmfs to be dyadic, problem~\eqref{prob:dyadic} is a convex optimization problem and can be solved efficiently. However, the restriction to dyadic pmfs makes the set of argument $\vecp$ discrete and the problem is not convex anymore. To the best of our knowledge, there is no efficient algorithm known that directly solves the problem. We therefore write the problem as a trade-off problem by adding a scaled version $\lambda\vecw^T\vecd$ of the shadowing to the objective function. This can be written as
\begin{align}
\kl(\vecd\Vert\vect)+\lambda\vecw^T\vecd&=\sum_i d_i\log\frac{d_i}{t_i}+\lambda\vecw^T\vecd\\
&=\kl(\vecd\Vert\vect\circ2^{-\lambda\vecw}).
\end{align}
The solution can efficiently be found by \emph{geometric Huffman coding} (\textsc{Ghc}), i.e., $\vecd=\textsc{Ghc}(\vect\circ2^{-\lambda\vecw})$. See \cite{Bocherer2011} for the definition of \textsc{Ghc} and \cite{website:ghc} for an implementation in Matlab. The shadowing constraint can be guaranteed by iteratively adapting $\lambda$: if for the resulting $\vecd$, $\vecw^T\vecd>S$, increase $\lambda$ and repeat, if $\vecw^T\vecd<S$, decrease $\lambda$ and repeat. Thus, the solution can be found by bisection. In summary, we have the following algorithm, which we call \emph{cost constrained geometric Huffman coding} (\textsc{ccGhc}).
\begin{algorithm}[(\textsc{ccGhc})]\ 
\\
$\ell<\lambda^*<u$\\
\textbf{repeat}\\
\indent 1. $\lambda=\frac{\ell+u}{2}$\\
\indent 2. $\vecd=\textsc{Ghc}(\vect\circ 2^{-\lambda\vecw})$\\
\indent 3. \textbf{if} $\vecw^T\vecd\leq S$, $u\leftarrow \lambda$; \textbf{else} $\ell\leftarrow \lambda$\\
\textbf{until} $u-\ell<\epsilon$\\
$\lambda^*=u$\\
$\vecd=\textsc{Ghc}(\vect\circ 2^{-\lambda^*\vecw})$
\end{algorithm}
\subsection{Asymptotic Achievability}
To evaluate the quality of the dyadic pmf found by \textsc{ccGhc}, we compare to what can be achieved when dropping the restriction to dyadic pmfs, i.e., when allowing the argument $\vecp$ in problem \eqref{prob:dyadic} to be any pmf from the probability simplex. Denote the optimal pmf from the probability simplex by $\vecp^*$. Since there is only a finite number of dyadic pmfs of a given length, the performance of the dyadic pmf $\vecd$ found by \textsc{ccGhc} may be too bad compared to what is achieved by the optimal pmf $\vecp^*$. This problem can be solved by generating dyadic pmfs of blocks of symbols. Consider the target pmf $\vect^k$ of $k$ consecutive symbols. The corresponding shadowing is given by the Kronecker sum $\vecv_k=\vecw^{\oplus k}$ of $k$ copies of $\vecw$. For an increasing blocklength, we have the following result. 
\begin{proposition}\label{prop:achievability}
Define $\vecd_k=\textsc{ccGhc}(\vect^k,\vecv_k,kS)$. Then, 
\begin{align}
\frac{\kl(\vecd_k\Vert\vect^{k})}{k}\to &\kl(\vecp^*\Vert\vect)\\
\text{and }\frac{\vecv_k^T\vecd_k}{k}\to & S,\;\frac{\vecv_k^T\vecd_k}{k}\leq S
\end{align}
i.e., the distance from the target pmf per symbol converges to the optimal value and the average shadowing per symbol converges to the target shadowing $S$, while the shadowing constraint is always fulfilled.
\end{proposition}
\begin{proof}
The proof is given in Section~\ref{sec:analysis}.
\end{proof}

\section{Writing to the ICT Cubes}
\begin{figure}
\footnotesize
\centering
\def\svgwidth{1.0\columnwidth}
\executeiffilenewer{images/system.svg}{images/system.pdf}%
{inkscape -z -D --file=images/system.svg %
--export-pdf=images/system.pdf --export-latex}%
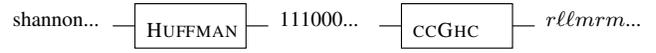%

\caption{We first compress the text to a binary sequence and then match the binary sequence to the design criteria by using \textsc{ccGhc}.}
\label{fig:system}
\end{figure}
\begin{table*}
\caption{The employed \textsc{Huffman} code.}
\label{tab:huffman}
\begin{tabular}{lllllllll}
\_ : 000 & a : 0100 & b : 101110 & c : 01101 & d : 11110 & e : 110 & f : 11111 & g : 001110 & h : 00110\\
i : 0101 & j : 001111111 & k : 00111101 & l : 01100 & m : 10110 & n : 1000 & o : 0111 & p : 100101&q : 001111110\\
r : 1010 & s : 1110 & t : 0010 & u : 10011 & v : 00111100 & w : 101111 & x : 001111100 & y : 100100 & z : 001111101
\end{tabular}
\end{table*}
\begin{table*}
\caption{The matching code induced by $\vecd_3=\textsc{ccGhc}(\vect^3,\vecv_3,3S)$.}
\label{tab:ccghc}
\begin{tabular}{rrrrrrrrr}
0010 : lll & 1101 : llr & 00000 : llm & 1100 : lrl & 1111 : lrr & 00011 : lrm & 00010 : lml & 01101 : lmr & 0000111 : lmm\\
1110 : rll & 1001 : rlr & 01100 : rlm & 1000 : rrl & 1011 : rrr & 01111 : rrm & 01110 : rml & 01001 : rmr & 000010 : rmm\\
01000 : mll & 01011 : mlr & 001101 : mlm & 01010 : mrl & 1010 : mrr & 001100 : mrm & 001111 : mml & 001110 : mmr & 0000110 : mmm
\end{tabular}
\end{table*}
\begin{figure*}
\centering
\includegraphics[scale=0.85]{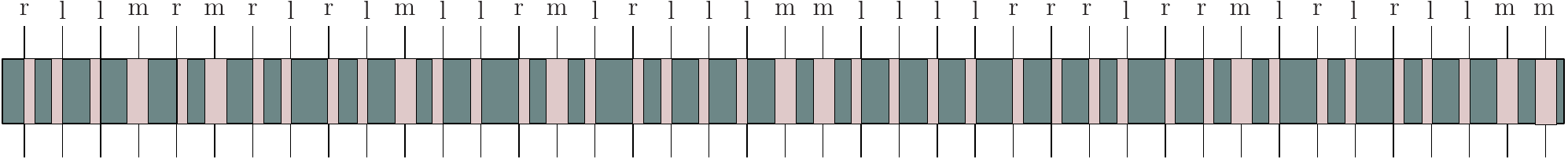}
\caption{Decoding the top floor with the codes specified in Table~\ref{tab:ccghc} and Table~\ref{tab:huffman} results in \texttt{shannon the fu}. This is the first part of \texttt{shannon the fundamental problem of communication is that of reproducing at one point either exactly or approximately a message selected at another point}, a phrase taken from the first chapter of \cite{Shannon1948}.}
\label{fig:topfloor_ICT}
\end{figure*}
We now apply \textsc{ccGhc} to solve the design problem of finding an encoding scheme subject to the design criteria C1.-C3. as stated in the introduction. The text that we write to the facades of the ICT cubes consists of quotes from scientists that significantly contributed to the development of information and communications technology. Our coding scheme consists of two parts. We first compress the text to a binary sequence by Huffman coding, and then match the binary sequence to the design criteria by using \textsc{ccGhc}. See Fig.~\ref{fig:system} for an illustration.
\subsection{Text Compression}
To keep the number of symbols small, we write the text using only small Latin characters and space, which results in an alphabet size of $27$. To map the text to a binary sequence, we use the Huffman code \cite{Huffman1952} of the relative symbol frequencies in the text. See Tab.~\ref{tab:huffman} for the resulting code. $49.4$\% of the bits in the resulting binary sequence are zeros and $50.6\%$ are ones, so roughly speaking, our assumption to have an iid sequence of equiprobable bits at the binary interface is reasonable.

\subsection{Criteria Matching}

We now map the binary sequence blockwise to a sequence of slats. The objective is to match the design criteria C1.-C3. as stated in the introduction. To see how close we are to the optimum, we calculate the optimal pmf $\vecp^*$ when the restriction to dyadic pmfs is dropped. The optimal pmf is given by
\begin{align}
\vecp^* = ( 0.3988,\,0.3988,\,0.2023)^T.
\end{align}
This is the pmf closest to the uniform pmf, thus the best match of criterion C1., while fulfilling the shadowing constraints C2. and C3. 

We choose $k=3$ as blocklength for the matcher codes. As a first matcher code, we use the code induced by the dyadic pmf $\vecd_3=\textsc{ccGhc}(\vect^3,\vecv_3,3S)$. The resulting code is displayed in Table~\ref{tab:ccghc}. The first row of the resulting sequence of slats is displayed in Fig.~\ref{fig:topfloor_ICT}. The interested reader is invited to decode it by using first the matching code in Table~\ref{tab:ccghc} in inverse direction and then the Huffman code in Table~\ref{tab:huffman} in inverse direction. The effective relative frequencies of the slats is
\begin{align}
\vecp_\mathrm{eff}&=\frac{1}{4264}(\sharp\{\text{left}\},\,\sharp\{\text{right}\},\,\sharp\{\text{middle}\})^T\\
&=(0.3838,\,0.39457,\,0.22162)^T.
\end{align}
As we can see, $\vecp_\mathrm{eff}$ is very close to $\vecp^*$. The effective shadowing is
\begin{align}
S_\mathrm{eff} =  0.20881,
\end{align}
which corresponds to an average shadowing of $33.4\%$. This exceeds the target percentage of $33\%$ by $0.4$ percentage points, thus violates criterion C3. This problem can be fixed as follows. We use a stricter shadowing constraint $S'=0.206$ instead of the original target constraint $S=0.2063$ and calculate $\vecd_3'=\textsc{ccGhc}(\vect^3,\vecv_3,3S')$. The effective relative slats frequencies that result from the code induced by $\vecd_3'$ are now
\begin{align}
\vecp_\mathrm{eff}'=(0.39132,\,0.4317,\,0.17698)^T.
\end{align}
and the effective shadowing is
\begin{align}
S_\mathrm{eff}'=0.20301.
\end{align}
This corresponds to an average shadowing of $32.5\%$, thus fulfills criterion C3. Note that $\vecp_\mathrm{eff}$ is closer to the uniform pmf than $\vecp_\mathrm{eff}'$ and thus matches better criterion C1. It is now up to the architects to choose among code $\vecd_3$ and $\vecd_3'$, i.e., to find the best trade-off between criterion C1. and the criteria C2. and C3. for their purpose.

\section{Analysis of \textsc{ccGhc}}
\label{sec:analysis}
This section consist of two parts. In Subsection~A, we derive two lemmas that characterize the operating point geometry in terms of average cost and distance to the target pmf. We then use these two lemmas in Subsection~B to actually prove Proposition~1.
\subsection{Operating Point Geometry}
We start by characterizing the region of achievable operating points. We define the \emph{distance-cost function} $\mathsf{D}(E)$ pointwise by the solution of 
\begin{align}
\begin{split}
\minimize_{\vecp}\quad& \kl(\vecp\Vert\vect)\\
\text{subject to}\quad
&\vecw^T\vecp-E\leq 0\\
&-\vecp\leq 0\\
&\vecone^T\vecp-1=0
\end{split}
\label{prob:simplex}
\end{align}
i.e., if $\vecp^*$ is the optimal pmf for $E=E^*$, then $\mathsf{D}(E^*)=\kl(\vecp^*\Vert\vect)$. Note that the two last constraints restrict $\vecp$ to the probability simplex, i.e., ensure that $\vecp$ is a pmf. By the convention $\log 0 = -\infty$, clearly, whenever $t_i=0$, the optimal pmf assigns $p_i^*=0$, since otherwise, the objective function would take the value infinity. Therefore, without loss of generality, we assume in the following that $t_i>0$ for all $i$. The Lagrangian is
\begin{align}
L(\vecp,\lambda,\pmb{\mu},\nu)=\kl(\vecp\Vert\vect)+\lambda(\vecw^T\vecp-E)-\pmb{\mu}^T\vecp+\nu(\vecone^T\vecp-1).
\end{align}
Assume $\vecp$ is feasible. Then the KKT conditions are
\begin{align}
\lambda\geq 0,\,
\vecmu&\geq 0\\
\lambda(\vecw^T\vecp-E)&=0\label{eq:kkt:weight}\\
\mu_ip_i&=0\label{eq:dmcCost:kktMu}\\
\frac{\partial L(\vecp,\vecmu,\nu,\lambda)}{\partial p_i}=\log\frac{p_i}{t_i}-1 -\mu_i+\nu+\lambda w_i&=0\label{eq:dmcCost:kktPartial}
\end{align}
It can be shown that for Problem~\eqref{prob:simplex}, a pmf $\vecp$ is optimal if and only if there are $\lambda,\pmb{\mu},\nu$ such that $\vecp$ fulfills the KKT conditions. Denote now by $\vecp^*,\lambda,\pmb{\mu},\nu$ values that fulfill the KKT conditions. By the last condition,
\begin{align}
\log p^*_i&=\log t_i + 1 + \mu_i-\nu-\lambda w_i
\end{align}
since by assumption $t_i>0$, the right-hand side is finite, therefore, $p_i>0$. Thus, by \eqref{eq:dmcCost:kktMu}, $\mu_i=0$ and we conclude
\begin{align}
\log p^*_i&=\log t_i + 1 -\nu-\lambda w_i,\qquad i=1,\dotsc,m.
\end{align}
\begin{lemma}\label{lem:convexity}
For $w_{\min}<E<\vecw^T\vect$, the distance-cost function $\mathsf{D}(E)$ is strictly convex in $E$.
\end{lemma}
\begin{proof}
Denote by $\vecp^*$ and optimal pmf for $E=E^*$. Since $\vecp^*$ is a pmf,
\begin{align}
p^*_i=\frac{p^*_i}{\sum_j p^*_j}&=\frac{t_ie^{1-\nu-\lambda w_i}}{\sum_j t_j e^{1-\nu-\lambda w_j}}\\
&=\frac{t_i e^{-\lambda w_i}}{\sum_j t_j e^{-\lambda w_j}}.
\end{align}
Since by assumption $E^*<\vecw^T\vect$, the average weight constraint is active, which implies $\lambda>0$. Thus, by \eqref{eq:kkt:weight}, $\vecw^T\vecp^*=E^*$, i.e,
\begin{align}
\vecw^T\vecp^* = \frac{\sum_i w_i t_i e^{-\lambda w_i}}{\sum_j t_j e^{-\lambda w_j}}\triangleq f(\lambda)=E^*.
\end{align}
We differentiate $f(\lambda)$ and get
\begin{align}
\frac{\de f(\lambda)}{\de \lambda}=\frac{\sum_i\sum_j (w_i w_j-w_i^2)t_i t_j e^{-\lambda (w_i+w_j)}}{\left[\sum_j t_j e^{-\lambda w_j}\right]^2}
\end{align}
We now want to show that $\frac{\de f(\lambda)}{\de \lambda}<0$. Since the denominator is positive, we only need to consider the numerator. We have
\begin{align}
&\sum_i \sum_j (w_iw_j-w_i^2)t_i t_j e^{-\lambda (w_i+w_j)}\\
&=\sum_i \sum\limits_{j\geq i}(w_iw_j-w_i^2+w_jw_i-w_j^2)t_i t_j e^{-\lambda (w_i+w_j)}\\
&=\sum_i \sum\limits_{j\geq i}[-(w_i-w_j)^2]t_i t_j e^{-\lambda (w_i+w_j)}<0
\end{align}
where the inequality in the last line follows since there is at least one pair $(i,j)$ such that $w_i\neq w_j$ and since, by assumption, $t_i>0$ for all $i$. Thus, $f$ is strictly monotonically decreasing and thereby invertible on its image, i.e., on $(w_{\min},\vecw^T\vect)$. Consequently, $\lambda=f^{-1}(E)$ is strictly monotonically decreasing. By \cite[Sec.~5.6.3]{Boyd2004},
$\lambda = -\frac{\de \mathsf{D}(E)}{\de E}$, 
thus,
\begin{align}
\frac{\de^2 \mathsf{D}(E)}{\de E^2}=-\frac{\de f^{-1}(E)}{\de E}>0
\end{align}
which shows the strict convexity of $\mathsf{D}(E)$ in $E$.
\end{proof}
\begin{figure*}[!t]
\footnotesize
\centering
\parbox{\columnwidth}{
\def\svgwidth{1.0\columnwidth}
\executeiffilenewer{images/geometryD.svg}{images/geometryD.pdf}%
{inkscape -z -D --file=images/geometryD.svg %
--export-pdf=images/geometryD.pdf --export-latex}%
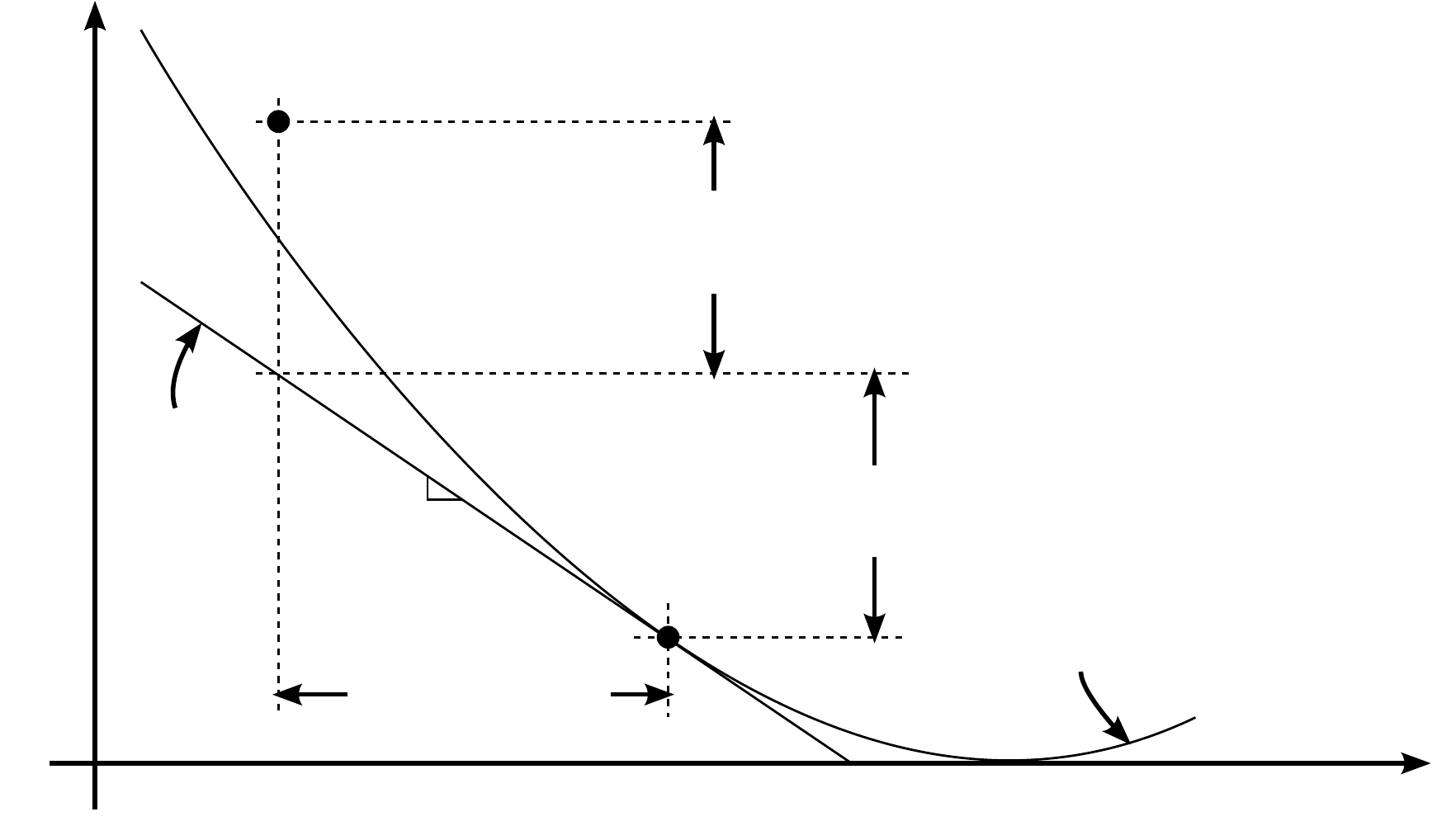%

\caption{}
\label{fig:geometry}
}
\parbox{\columnwidth}{
\def\svgwidth{1.0\columnwidth}
\executeiffilenewer{images/achievabilityD.svg}{images/achievabilityD.pdf}%
{inkscape -z -D --file=images/achievabilityD.svg %
--export-pdf=images/achievabilityD.pdf --export-latex}%
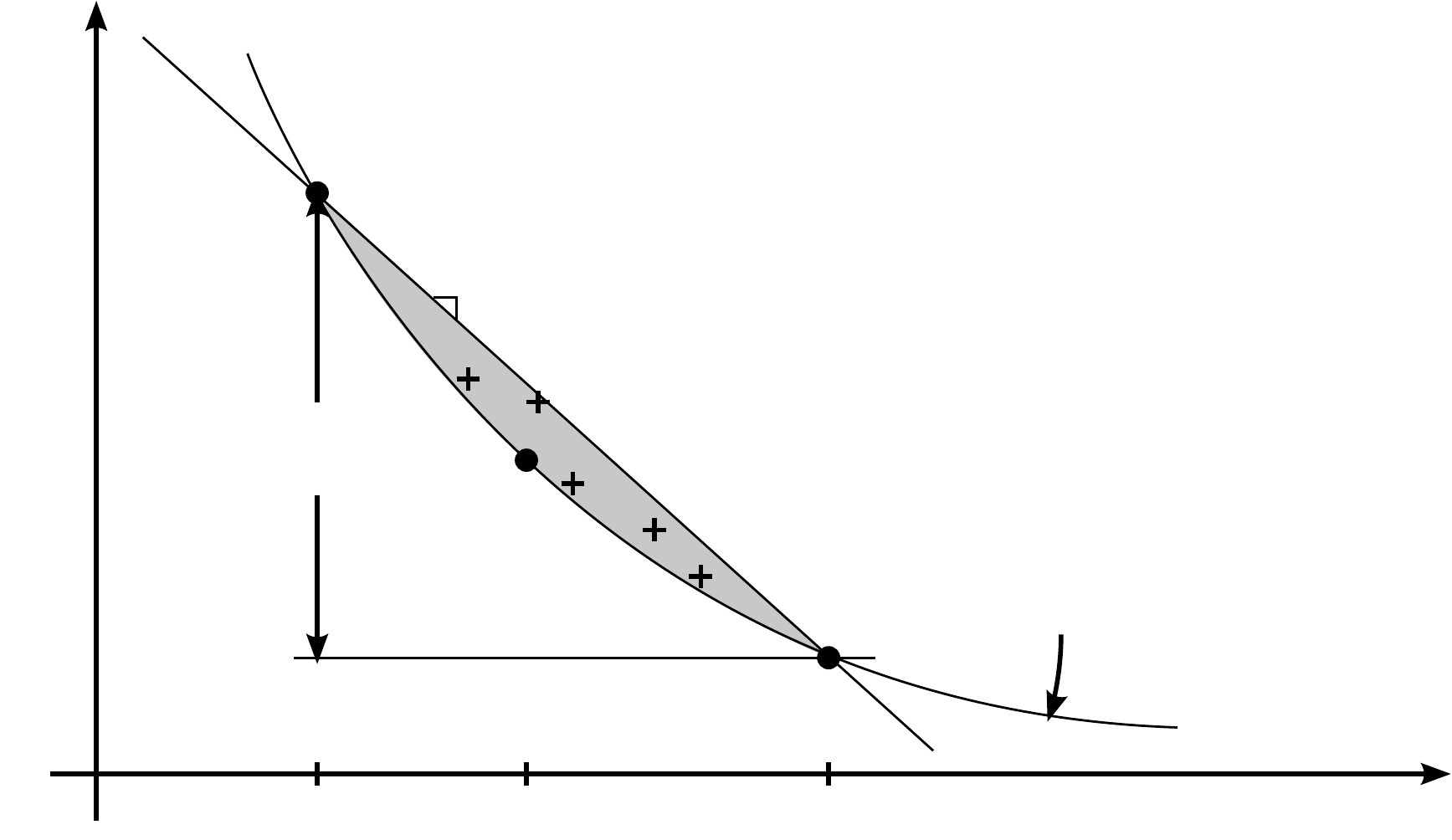%

\caption{}
\label{fig:achievability}
}
\end{figure*}

\begin{lemma}\label{lem:geometry} 
For a given cost constraint $E^*$, denote by $\vecp^*$ an optimal pmf. Denote by $\vecp$ an arbitrary pmf with the only restriction that $p_i=0$ whenever $p_i^*=0$. Then
\begin{align}
\kl(\vecp\Vert\vect)=\mathsf{D}(E^*)-\lambda(\vecw^T\vecp-E^*)+\kl(\vecp\Vert\vecp^*)\label{eq:geometry}.
\end{align}
where $-\lambda$ is the the slope of the tangent of $\mathsf{D}$ in $(E^*,\mathsf{D}(E^*))$.
\end{lemma}
\begin{proof}
\begin{align}
\kl(\vecp\Vert\vect)&=\sum_i p_i\log\frac{p_i}{t_i}\\
&=\sum_i p_i\log\frac{p_ip_i^*}{t_ip_i^*}\\
&=\sum_i p_i\log\frac{p_i^*}{t_i}+\kl(\vecp\Vert\vecp^*)\\
&=\sum_i p_i\log p_i^*-\sum_i p_i\log t_i+\kl(\vecp\Vert\vecp^*)
\end{align}
We further develop the first term
\begin{align}
\sum_i &p_i\log p_i^*=\sum_i (p_i+p_i^*-p_i^*)\log p_i^*\\
&=-\entop(\vecp^*)+\sum_i (p_i-p_i^*)\log p_i^*\\
&=-\entop(\vecp^*)+\sum_i (p_i-p_i^*)(\log t_i+1-\nu-\lambda w_i)\\
&=-\entop(\vecp^*)-\lambda(\vecw^T\vecp-\vecw^T\vecp^*)+\sum_i p_i\log t_i\nonumber\\
&\qquad\qquad-\sum_i p_i^*\log t_i\\
&=\kl(\vecp^*\Vert\vect)-\lambda(\vecw^T\vecp-\vecw^T\vecp^*)+\sum_i p_i\log t_i.
\end{align}
All together,
\begin{align}
\kl(\vecp\Vert\vect)&=\kl(\vecp^*\Vert\vect)-\lambda(\vecw^T\vecp-\vecw^T\vecp^*)+\kl(\vecp\Vert\vecp^*)\\
&=\mathsf{D}(E^*)-\lambda(\vecw^T\vecp-E^*)+\kl(\vecp\Vert\vecp^*).
\end{align}
\end{proof}
\subsection{Proof of Proposition~\ref{prop:achievability}}
We now show that any target operating point $Q^*=(\vecw^T\vecp^*,\mathsf{D}(\vecw^T\vecp^*))$ can be achieved by a dyadic pmf. We do this in two steps. First, we show the existence of dyadic operating points close to the target operating point, and then we show that \textsc{ccGhc} actually finds them. Both results are a direct consequence of the strict convexity of the distance-cost function $\mathsf{D}(E)$ that we stated in Lemma~\ref{lem:convexity}.
\subsubsection{Existence of good dyadic points}
\label{subsec:existence}

Consider the optimal pmf $\vecp^{*k}$ of $k$ consecutive symbols. Define $\vecv_k=\vecw^{\oplus k}$ where $\vecw^{\oplus k}$ denotes the Kronecker sum of $k$ copies of $\vecw$. Furthermore, define $\vecd_k=\textsc{Ghc}(\vecp^{*k})$. By Lemma~\ref{lem:geometry}, the operating point geometry becomes
\begin{align}
\frac{\kl(\vecd_k\Vert\vect^k)}{k} = \mathsf{D}(E^*) - \lambda\Bigl(\frac{\vecv_k^T\vecd_k}{k}-E^*\Bigr)+\frac{\kl(\vecd_k\Vert\vecp^{*k})}{k}\label{eq:kgeometry}.
 \end{align}
By \cite[Prop.~2]{Bocherer2011}, since $\vecd_k=\textsc{Ghc}(\vecp^{*k})$, the normalized KL-distance on the right-hand side goes to zero as $k\to\infty$. Consider now Fig.~\ref{fig:geometry}. The tangent of $\mathsf{D}(E)$ in $Q^*$ is given by
\begin{align}
g(E)=\mathsf{D}(E^*) - \lambda(E-E^*).
\end{align}
As the normalized KL-distance of $\vecd_k$ to $\vecp^{*k}$ gets smaller, the normalized KL-distance of $\vecd_k$ to $\vect^k$ on the left-hand side of \eqref{eq:kgeometry} is approaching the tangent $g$. However, because the tangent is linear in $E$ and $\mathsf{D}$ is strictly convex and lower bounds $\frac{\kl(\vecd_k\Vert\vect^{*k})}{k}$, the dyadic operating point $(\frac{\vecv_k^T\vecd_k}{k},\frac{\kl(\vecd_k\Vert\vect^{*k})}{k})$ has to approach $Q^*$ both in terms of distance and cost.

\subsubsection{Finding good dyadic points}

It remains to show that algorithm \textsc{ccGhc} finds good dyadic points. This can best be seen in Fig.~\ref{fig:achievability}. Suppose we want to find a dyadic pmf $\vecd_k$ such that for a given $\epsilon > 0$,
\begin{align}
\frac{\kl(\vecd_k\Vert\vect^*)}{k}\leq\mathsf{D}(E^*)+\epsilon\text{ and }\frac{\vecv_k^T\vecd_k}{k}\leq E^*.\label{eq:requirements}
\end{align}
Define
\begin{align}
E' : \mathsf{D}(E')=\mathsf{D}(E^*)+\epsilon\text{ and }E''=\frac{E'+E^*}{2}.
\end{align}
The chord from $Q^*=(E^*,\mathsf{D}(E^*))$ to $Q'=(E',\mathsf{D}(E'))$ cuts a segment from the area above $\mathsf{D}$. Because of the strict convexity of $\mathsf{D}$, this segment is nonempty. Note that all operating points in the segment fulfill the requirements \eqref{eq:requirements}. As shown in the previous Subsection~\ref{subsec:existence}, for a big enough $k$, there are dyadic operating points approximating $Q''=(E'',\mathsf{D}(E''))$ that lie within this segment. Define now $-\xi$ as the slope of the chord, i.e.,
\begin{align}
\xi = -\frac{\mathsf{D}(E')-\mathsf{D}(E^*)}{E'-E^*}.
\end{align}
Now, $\vecd_k=\textsc{Ghc}(\vect^k\circ 2^{-\xi\vecv_k})$ minimizes
\begin{align}
\frac{1}{k}\Bigl[\kl(\vecd_k\Vert\vect^k)+\xi\vecv_k^T\vecd\Bigr]
\end{align}
and will thus find a point in the segment. The slope $-\xi$ will also be evaluated by \textsc{cCghc}, thus $\vecd_k=\textsc{ccGhc}(\vect^k,\vecv_k,kE^*)$ will give a dyadic operating point at least as good as 
\begin{align}
\vecd_k=\textsc{Ghc}(\vect^k\circ 2^{-\xi\vecv_k}).
\end{align}
This concludes the proof of Proposition~1.

\bibliographystyle{IEEEtran}
\normalsize
\bibliography{IEEEabrv,confs-jrnls,Dnc}

\end{document}